\documentclass[11pt,reqno]{amsart}
\usepackage{amscd,amssymb,amsmath,amsthm}
\usepackage[arrow,matrix]{xy}
\usepackage{graphicx}
\usepackage{cite}
\newtheorem{thm}{Theorem}

\newtheorem{con}{Condition}
\newtheorem{rk}{Remark}

\newtheorem{ex}{Example}

\numberwithin{equation}{section}

\newcommand{\s}{{\sigma}}
\newcommand{\de}{{\xi}}

\newcommand{\bea}{\begin{eqnarray}}
\newcommand{\eea}{\end{eqnarray}}

\begin{document}
\title[Non-translation-invariant Gibbs measure]{Non-translation-invariant Gibbs Measures
for Models With Uncountable Set of Spin Values on a Cayley Tree}

\author{U.\ A.\ Rozikov, G.\ I.\ Botirov}

  \address{U.\ A.\ Rozikov and G.I.Botirov\\Institute of mathematics,
81, Mirzo Ulug'bek str., 100125, Tashkent, Uzbekistan.}
\email {rozikovu@yandex.ru \ \ botirovg@yandex.ru}

\begin{abstract} We consider models with nearest-neighbor
interactions and with the set $[0,1]$ of spin values, on a Cayley
tree of order $k\geq 1$.
 It is known that the "splitting Gibbs measures" of the model
can be described by solutions of a nonlinear integral
equation. Recently, solving this integral equation some periodic
(in particular translation-invariant) splitting Gibbs measures were found. In
this paper we give three constructions of new sets of non-translation-invariant splitting Gibbs measures.
Our constructions are based on known solutions of the integral equation.
\end{abstract}
\maketitle



\section{Introduction} \label{sec:intro}

Let us first give necessary definitions, then explain what is the main problem;
secondly we give the history of its solutions and then formulate the part of the
problem which we want to solve in this paper.

A {\it Cayley tree} $\Gamma^k$
of order $ k\geq 1 $ is an infinite tree, i.e., a graph without
cycles, such that exactly $k+1$ edges originate from each vertex.
Let $\Gamma^k=(V, L)$ where $V$ is the set of vertices and  $L$ the set of edges.

Two vertices $x$ and $y$ are called {\it nearest neighbors} if there exists an
edge $l \in L$ connecting them.
We will use the notation $l=\langle x,y\rangle$.

A collection of nearest neighbor pairs $\langle x,x_1\rangle, \langle x_1,x_2
\rangle,...,\langle x_{d-1},y\rangle$ is called a {\it
path} from $x$ to $y$. The {\it distance} $d(x,y)$ on the Cayley tree
is the number of edges of the shortest path from $x$ to $y$.

For a fixed $x^0\in V$, called the {\it root}, we set
\begin{equation*}
W_n=\{x\in V\,| \, d(x,x^0)=n\}, \ \ \ V_n=\bigcup_{m=0}^n W_m,  \ \ \ L_n=\{\langle x,y\rangle\in L: x,y\in V_n\}
\end{equation*}
and denote
$$
S_k(x)=\{y\in W_{n+1} :  d(x,y)=1 \}, \ \ x\in W_n, $$
the set  of {\it direct successors} of $x$ on the Cayley tree of order $k$.

We consider models where the spin takes values in the set $[0,1]$, and spins are assigned to the vertices
of the tree. For $A\subset V$ a {\it configuration} $\s_A$ on $A$ is an arbitrary function $\s_A:A\to
[0,1]$. Denote $\Omega_A=[0,1]^A$ the set of all configurations on $A$. We denote $\Omega=[0,1]^V$.

The {\it Hamiltonian} of the model is :
\begin{equation}\label{e1}
 H(\sigma)=-J\sum_{\langle x,y\rangle\in L}
\de_{\sigma(x)\sigma(y)},
\end{equation}
where $J \in R\setminus \{0\}$
and $\de: (u,v)\in [0,1]^2\to \de_{uv}\in R$ is a given bounded,
measurable function.

Let $\lambda$ be the Lebesgue measure on $[0,1]$.  On the set of all
configurations on $A$ the a priori measure $\lambda_A$ is introduced as
the $|A|$-fold product of the measure $\lambda$, where
$|A|$ denotes the cardinality of $A$.

 We consider a standard
sigma-algebra ${\mathcal B}$ of subsets of $\Omega=[0,1]^V$ generated
by the measurable cylinder subsets.

 A probability measure $\mu$ on $(\Omega,{\mathcal B})$
is called a {\it Gibbs measure} (corresponding to the Hamiltonian $H$) if it satisfies the DLR equation, namely for any
$n=1,2,\ldots$ and $\sigma_n\in\Omega_{V_n}$:
$$\mu\left(\left\{\sigma\in\Omega :\;
\sigma\big|_{V_n}=\sigma_n\right\}\right)= \int_{\Omega}\mu ({\rm
d}\omega)\nu^{V_n}_{\omega|_{W_{n+1}}} (\sigma_n),$$
where
$$ \nu^{V_n}_{\omega|_{W_{n+1}}}(\sigma_n)=\frac{1}{Z_n\left(
\omega\big|_{W_{n+1}}\right)}\exp\;\left(-\beta H
\left(\sigma_n\,||\,\omega\big|_{W_{n+1}}\right)\right),
$$
and $\beta={1\over T}$, $T>0 $ is temperature.
Furthermore, $\sigma\big|_{V_n}$ and $\omega\big|_{W_{n+1}}$ denote the
restrictions of configurations $\sigma,\omega\in\Omega$ to $V_n$ and $W_{n+1}$, respectively. Next,
$\sigma_n:\;x\in V_n\mapsto \sigma_n(x)$ is a configuration in $V_n$ and
$$H\left(\sigma_n\,||\,\omega\big|_{W_{n+1}}\right)=-J\sum_{\langle x,y\rangle\in L_n}\de_{\sigma_n(x)\sigma_n(y)}-J\sum_{\langle x,y\rangle:\;x\in V_n, y\in W_{n+1}}
\de_{\sigma_n(x)\omega (y)}.$$
Finally,
$$Z_n\left(\omega\big|_{W_{n+1}}\right)=
\int_{\Omega_{V_n}} \exp\;\left(-\beta H \left({\widetilde\sigma}_n\,||\,\omega
\big|_{W_{n+1}}\right)\right)\lambda_{V_n}(d{\widetilde\sigma}_n).$$

 The {\it main problem} for a given Hamiltonian is to describe all its Gibbs measures.
 See \cite{Ge} for a general definition of Gibbs measure, motivations why these measures are
important and the theory of such measures.

 This main problem is not completely solved even for simple Ising or Potts models
 on a Cayley tree with a finite set of spin values.
 Mainly this problem is solved for the class of splitting Gibbs measures (SGMs) \cite{R} (Markov chains \cite{Ge}),
 which are limiting Gibbs measures constructed by Kolmogorov's
 extension theorem of the following finite-dimensional distributions:
 given $n=1,2,\ldots$, consider the probability distribution $\mu_n$ on $\Omega_{V_n}$ defined by
\begin{equation}\label{uy2}
\mu_n(\sigma_n)=Z_n^{-1}\exp\left(-\beta H(\sigma_n)
+\sum_{x\in W_n}h_{\sigma(x),x}\right),
\end{equation}
where $h:\;x\in V\mapsto h_x=(h_{t,x}, t\in [0,1]) \in R^{[0,1]}$ be mapping of $x\in V\setminus
\{x^0\}$. Here $Z_n$ is the corresponding partition function.
The probability distributions $\mu_n$ are compatible if for any $n\geq 1$ and
$\sigma_{n-1}\in\Omega_{V_{n-1}}$:
\begin{equation}\label{uy4}
\int_{\Omega_{W_n}}\mu_n(\sigma_{n-1}\vee\omega_n)\lambda_{W_n}(d(\omega_n))=
\mu_{n-1}(\sigma_{n-1}).
\end{equation}
 Here
$\sigma_{n-1}\vee\omega_n\in\Omega_{V_n}$ is the concatenation of
$\sigma_{n-1}$ and $\omega_n$.

To see that a SGM satisfies the DLR equation, we consider any finite volume $D$
and note that for any finite $n$ which is sufficiently large we have
\begin{equation}
\mu_n\left(\left\{\omega\in\Omega:\;
\omega\big|_{D}=\sigma_D\right\}\right)= \int_{\Omega}\mu_n ({\rm d}\varphi)\nu^{D}_\varphi
(\sigma_D),
\end{equation}
which follows from the compatibility property of the finite-volume Gibbs measures.

 For the model (\ref{e1}) on the Cayley tree, in \cite{re}, the problem of describing the SGMs was reduced to the description of the solutions of the following
integral equation
\begin{equation}\label{e5}
 f(t,x)=\prod_{y\in S_k(x)}{\int_0^1K(t,u)f(u,y)du \over \int_0^1K(0,u)f(u,y)du}.
 \end{equation}
Here, $K(t,u)=\exp(J\beta\de_{tu})$ and the unknown function is $f(t,x)>0, \ x\in V, t\in [0,1]$ and
$du=\lambda(du)$ is the Lebesgue measure.

If a solution $f(t,x)$ is given then the corresponding SGM $\mu$
on $\Omega$ is such that, for any $n$ and
$\sigma_n\in\Omega_{V_n}$,
$$\mu \left(\left\{\sigma
\Big|_{V_n}=\sigma_n\right\}\right)=Z_n^{-1}\exp\left(-\beta H(\sigma_n)
+\sum_{x\in W_n}\ln f(\sigma(x),x)\right).
$$

A splitting Gibbs measure is called {\it translation-invariant measure} if it corresponds to a solution
$f(t,x)$ which does not depend on $x\in V$, i.e., $f(t,x)=f(t)$ for any $x\in V$.

In this paper we only deal with splitting Gibbs measures, therefore we omit the word ``splitting" in the following text.

Note that the analysis of solutions to (\ref{e5}) is
not easy. This difficulty depends on the given function $\xi$ (i.e. on $K(t,u)>0$).

Let us list known results about solutions of (\ref{e5}) and the Gibbs measures corresponding to them:

 In \cite{re} for $k=1$ it was shown that the integral equation has a
unique solution. In case $k\geq 2$ some models (with the set
$[0,1]$ of spin values) which have a unique Gibbs
measure are constructed.

In \cite{EHR12} several models with nearest-neighbor interactions and
with the set $[0, 1]$ of spin values, on a Cayley tree of order $k\geq 2$ are constructed. It is proved that each of the
constructed models has at least two translational-invariant Gibbs measures, i.e. the equation (\ref{e5}) has at least two solutions $f(t,x)$ which are independent of the vertices $x$ of the tree.

In \cite{EHR13} a condition on $K(t,u)$ is found under which the corresponding integral
equation (\ref{e5}) has a unique solution independent of $x$ (i.e. uniqueness of the
translation-invariant Gibbs measure).

In \cite{ERB} for a specifically chosen $K(t,u)$ it is shown that under certain conditions on the parameters of
the model there are at least two translation-invariant Gibbs measures (i.e., there are phase transitions).

In \cite{JKB} the authors considered a model on a Cayley tree of order $k\geq 2$,
where the function $\xi$ depends on a parameter $\theta\in [0,1)$. It is show that for $\theta \in [0, {5\over 3k}]$ the model has a unique translation-invariant Gibbs measure. If $\theta\in ({5\over3k}, 1)$ there is a phase transition, in particular there are three translation-invariant Gibbs measures.

 Paper  \cite{RH} deals with a class of Gibbs measures which are periodic and also a Markov chain.
 It is shown that the period must be either 1 or 2. If $k=1$ or the interaction is weak enough,
 the period is 1, i.e., every such Gibbs measure is translation-invariant. For $k=2$,
 a class of interactions is constructed admitting at least two Gibbs measures with period 2.
 For $k$ sufficiently large, an interaction is given admitting at least four Gibbs measures with period two.

 In \cite{ENH} the translation-invariant Gibbs measures for a function $K(t,u)$ are investigated by properties
 of positive fixed points of quadratic operators. Under some conditions it is shown that there are two and three positive fixed points.

We note that in the above mentioned papers existence of a Gibbs measure
is proved by directly solving the equation (\ref{e5})
for concrete chosen $K(t,u)$.

In this paper our {\it aim} is slightly different:
we mainly will construct new solutions of (\ref{e5}) by its known solutions.
To do this we will adapt to our models the construction methods which were
used for models with a {\it finite} set of spin values (see \cite{ART}, \cite{BG}, \cite{Ge}, \cite{R}).

\section{Non-translation-invariant Gibbs measures}

\subsection{ART construction}
In \cite{ART} for the Ising model (with the set $\{-1,1\}$ of spin values)
the authors constructed a class of new Gibbs measures by extending the known Gibbs measures
defined on a Cayley tree of order $k_0$ to a Cayley tree of higher order $k>k_0$.
Their construction is called ART-construction \cite{GRRR}.

In this subsection we adapt the ART-construction to models with an uncountable set of spin values.

For a given $H(\sigma)$ of the model (\ref{e1}), denote by $\mathcal G_k(H)$ the set of
{\it all} splitting Gibbs measures on the Cayley tree of order $k\geq 2$.

By $|M|$ we denote the number of elements of a set $M$.

The main result of this subsection is the following

\begin{thm}\label{t1} Take $k_0,k\in \{2,3,\dots\}$ such that $k>k_0$. If $|\mathcal G_{k_0}(H)|\geq 2$
and $K(t,u)$ is such that
\begin{equation}\label{K0}
\int_0^1(K(t,u)-K(0,u))du=0, \ \ \forall t\in [0,1]
\end{equation}
then for each $\mu\in \mathcal G_{k_0}(H)$ there is a $\nu=\nu(\mu)\in \mathcal G_k(H)$.
\end{thm}
\begin{proof}  By our assumptions we have that ${\mathcal G}_{k_0}(H)$ contains at last two elements.
Condition (\ref{K0}) guarantees that $f(t,x)\equiv 1$ is a solution of the equation (\ref{e5}) for {\it any} $k\geq 2$.
Denote by $\mu_1$ the Gibbs measure which corresponds to this solution.

Now for any $\mu\in {\mathcal G}_{k_0}(H)\setminus \{\mu_1\}$,  we shall construct a
Gibbs measure $\nu=\nu(\mu)$ which is a measure on the Cayley tree of order $k>k_0$.
As mentioned in the previous section to each measure $\mu\in {\mathcal G}_{k_0}(H)$ corresponds a
unique function $f(t,x)=f_\mu(t,x)$ which satisfies (\ref{e5}) on $\Gamma^{k_0}$.
Construct a function $g(t,x)\equiv g_\mu(t,x)$ on $\Gamma^k$ as follows.
Let $V^k$ be the set of all vertices of the Cayley tree $\Gamma^k$.
Since $k_0<k$ one can consider $V^{k_0}$ as a subset of $V^k$.
Define the following function
\begin{equation}\label{1}
g(t,x)=\left\{\begin{array}{ll}
f_\mu(t,x), \ \ \mbox{if} \ \ x\in V^{k_0}\\[2mm]
1, \ \ \mbox{if} \ \ x\in V^k\setminus V^{k_0}.\\
\end{array}\right.
\end{equation}

Now we shall check that (\ref{1}) satisfies (\ref{e5}) on $\Gamma^k$.

Let $x\in V^{k_0}\subset V^k$. We have
$$ g(t,x)=\prod_{y\in S_k(x)}{\int_0^1K(t,u)g(u,y)du \over \int_0^1K(0,u)g(u,y)du}$$
$$=\prod_{y\in S_k(x)\cap V^{k_0}}{\int_0^1K(t,u)f_\mu(u,y)du \over \int_0^1K(0,u)f_\mu(u,y)du}
\prod_{y\in S_k(x)\cap (V^k\setminus V^{k_0})}{\int_0^1K(t,u)du \over \int_0^1K(0,u)du}$$
For the first product we use $S_k(x)\cap V^{k_0}=S_{k_0}(x)$ and for the second product we use the condition (\ref{K0}) then we get
$$ g(t,x) =\prod_{y\in S_{k_0}(x)}{\int_0^1K(t,u)f_\mu(u,y)du \over \int_0^1K(0,u)f_\mu(u,y)du}=f_\mu(t,x).$$

If $x\in V^k\setminus V^{k_0}$ then $S_k(x)\subset V^k\setminus V^{k_0}$. Therefore $g(u,y)=1$, for any $y\in S_k(x)$ and we have
$$g(t,x)=\prod_{y\in S_k(x)}{\int_0^1K(t,u)du \over \int_0^1K(0,u)du}=1.$$
 Thus $g(t,x), x\in V^k$ satisfies
the integral equation (\ref{e5}) and we denote by $\nu=\nu(\mu)$ the Gibbs measure which
corresponds to $g(t,u)$. By the construction one can see that $\nu(\mu_1)\ne \nu(\mu_2)$ if $\mu_1\ne \mu_2$ and the measure $\nu$ is
not translation-invariant.
\end{proof}

Now let us give some examples where the conditions of Theorem \ref{t1} are satisfied:

\begin{ex} Let  $k=2$. In the model (\ref{e1}) take
$$\xi_{tu}=\frac{1}{\beta J}\ln\left(1+\frac{14}{15}\cdot\sqrt[5]{4\left(t-\frac{1}{2}\right)\left(u-\frac{1}{2}\right)}\right), \,\ t,u\in[0,1].$$
Then, for the kernel $K(t,u)$ of (\ref{e5}) we have
$$K(t,u)=1+\frac{14}{15}\cdot\sqrt[5]{4\left(t-\frac{1}{2}\right)\left(u-\frac{1}{2}\right)}.$$
In \cite{EHR12} it was shown that this model has at least two Gibbs measures and
the condition (\ref{K0}) is satisfied, i.e., $f(t,x)\equiv 1$ is a solution to (\ref{e5}).
\end{ex}
\begin{ex}
Consider the case $k=3$ and
$$K(t,u)=1+\frac{1}{2}\sqrt[7]{4\left(t-\frac{1}{2}\right)\left(u-\frac{1}{2}\right)}.$$
This model also satisfies the conditions of Theorem \ref{t1} and has at least two Gibbs measures (see \cite{EHR12}).
\end{ex}
For other examples satisfying conditions of Theorem \ref{t1} see \cite{EHR12}, \cite{ERB}, \cite{RH}.

\subsection{Bleher-Ganikhodjaev construction}
Here we will adapt the Bleher-Ganikhodjaev construction of \cite{BG} for the model (\ref{e1}).

If an arbitrary edge $\langle x^0, x^1\rangle=l\in L$ is deleted from the Cayley tree $\Gamma^k$, it splits into two components -- two semi-infinite (half) trees $\Gamma^k_0$ and $\Gamma^k_1$.
Consider the half tree $\Gamma^k_0$, and denote by $V^0$ the set of its vertices. Namely the root $x^0$
has $k$ nearest neighbors.

Denoting $h(t,x)=\ln f(t,x)$ write the equation (\ref{e5}) as
\begin{equation}\label{eh}
 h(t,x)=\sum_{y\in S_k(x)}\ln{\int_0^1K(t,u)e^{h(u,y)}du \over
 \int_0^1K(0,u)e^{h(u,y)}du}.
 \end{equation}

On the set $C[0,1]$ of continuous functions define the following non-linear operator
\begin{equation}\label{Af}
Af(t)=\ln{\int_0^1K(t,u)e^{f(u)}du \over
 \int_0^1K(0,u)e^{f(u)}du},
\end{equation}
 where $K(t,u)>0$.
\begin{con}\label{c1} Assume $K(t,u)>0$ is continuous on $[0,1]^2$, i.e., $K(\cdot,\cdot)\in C^+[0,1]^2$, and there is $\alpha\equiv \alpha_K \in [0, 1)$ such that
$$|Af(t)-Ag(t)|\leq \alpha |f(t)-g(t)|, \ \ \forall f,g\in C[0,1], \forall t\in [0,1].$$
\end{con}
\begin{con}\label{c2} Assume there are at least two translation-invariant solutions, say $h(t,x)\equiv h(t)\in C[0,1]$ and $h(t,x)\equiv \eta(t)\in C[0,1]$, to the equation (\ref{eh}), i.e., they are fixed points for the operator $kA$:
\begin{equation}\label{ehh}
 h(t)=kAh(t)=k\ln{\int_0^1K(t,u)e^{h(u)}du \over
 \int_0^1K(0,u)e^{h(u)}du}, \ \ \ \
 \eta(t)=kA\eta(t).
 \end{equation}
 \end{con}
 \begin{rk} If Condition \ref{c1} is satisfied then to satisfy the Condition \ref{c2}
 it is necessary that ${1\over k}\leq\alpha<1$.
  \end{rk}

We use $h(t)$ and $\eta(t)$ to construct an uncountable set of new solutions of (\ref{eh}).

Consider an infinite path $\pi=\{x^0=x_0<x_1<\dots\}$  (the notation $x<y$ meaning that  paths from the root to $y$ go through $x$).
Associate  to this  path
a collection $h^\pi=\{h^\pi_{t,x}: x\in V^0, t\in [0,1]\}$ given by
\begin{equation}\label{b3.1}
h_{t,x}^\pi=\left\{\begin{array}{ll}
h(t), \ \ \mbox{if} \ \ x\prec x_n, \, x\in W_n,\\[2mm]
\eta(t), \ \ \mbox{if} \ \ x_n\prec x, \, x\in W_n,\\[2mm]
h_{t,x_n}, \ \ \mbox{if} \ \ x=x_n.
\end{array}\right.
\end{equation}
$n=1,2,\dots$ where $x\prec x_n$ (resp. $x_n\prec x$) means that $x$ is on the left (resp. right) from the path $\pi$
and $h_{t,x_n}$ are arbitrary numbers, some conditions on these numbers will be given below.

\begin{thm}\label{t2} If Conditions \ref{c1} and \ref{c2} are satisfied, then for any infinite path $\pi$, there exists a unique set of numbers
$h^\pi= \{h^\pi_{t,x}\}$ satisfying equations (\ref{eh}) and (\ref{b3.1}).
\end{thm}
\begin{proof} On $W_n$, we define the set
\begin{equation}\label{b4.1}
h_{t,x}^{(n)}=\left\{\begin{array}{lll}
h(t), \ \ \mbox{if} \ \ x\prec x_n, \, x\in W_n,\\[2mm]
\eta(t), \ \ \mbox{if} \ \ x_n\prec x, \, x\in W_n,\\[2mm]
h^{(n)}_{t,x},   \ \ \mbox{if} \ \ x=x_n,
\end{array}\right.
\end{equation}
where $h^{(n)}_{t,x_n}\in(h^{\min}(t), h^{\max}(t))$, $\forall t\in [0,1]$ is an arbitrary number.

 We extend the definition of $h^{(n)}_{t,x}$
for all $x\in V_n=\cup_{m=0}^nW_m$ using recursion equations (\ref{eh})
and prove that the limit
\begin{equation}\label{b4.2}
h_{t,x} = \lim_{n\to\infty} h_{t,x}^{(n)}
\end{equation}
exists for every fixed $x\in V^0$ and is independent of the choice of $h^{(n)}_{t,x}$ for $x = x_n$.

If $x\in W_{n-1}$ and $x\prec x_{n-1}$, then for any $y\in S_k(x)$ we have $y\prec x_{n}$, therefore
$$h_{t,x}^{(n)}=\sum_{y\in S_k(x)} \ln{\int_0^1K(t,u)e^{h^{(n)}_{u,y}}du \over
 \int_0^1K(0,u)e^{h^{(n)}_{u,y}}du}=kAh(t)=h(t).$$
Similarly, for $x\in W_{n-1}$ and $x_{n-1}\prec x$, we get $h_{t,x}^{(n)}=\eta(t)$.
Consequently, for any $x\in W_m$, $m\leq n$ we have
\begin{equation}\label{b3.1a}
h_{t,x}^{(n)}=\left\{\begin{array}{ll}
h(t), \ \ \mbox{if} \ \ x\prec x_m, \, x\in W_m,\\[2mm]
\eta(t), \ \ \mbox{if} \ \ x_m\prec x, \, x\in W_m.
\end{array}\right.
\end{equation}
This implies that limit (\ref{b4.2}) exists for $x\in W_m$, $x\ne x_m$ and
$$h_{t,x}=\left\{\begin{array}{ll}
h(t), \ \ \mbox{if} \ \ x\prec x_m, \, x\in W_m,\\[2mm]
\eta(t), \ \ \mbox{if} \ \ x_m\prec x, \, x\in W_m.
\end{array}\right.
$$
Therefore, we only need to
establish that limit (\ref{b4.2}) exists for $x = x_m$.
For $1\leq l\leq n$ we have
\begin{equation}\label{b4.3}
h_{t,x_{l-1}}^{(n)}=\sum_{y\in S_k(x_{l-1})}\ln{\int_0^1K(t,u)e^{h^{(n)}_{u,y}}du \over
 \int_0^1K(0,u)e^{h^{(n)}_{u,y}}du}=\sum_{y\in S_k(x_{l-1})}Ah^{(n)}_{t,y}.
\end{equation}
Consider two sets $\{\bar h_{t,x}^{(n)}, x\in V_n\}$ and $\{\tilde h_{t,x}^{(n)}, x\in V_n\}$ which correspond to two values
$\bar h^{(n)}_{t,x}$ and $\tilde h^{(n)}_{t,x}$ for $x = x_n$, in (\ref{b4.1}), then since $\tilde h_{t,y}^{(n)}=\bar h_{t,y}^{(n)}$, $\forall t\in [0,1]$ and for any $y\ne x_l$, $y\in W_l$, from (\ref{b4.3}) we get
\begin{equation}\label{b4.4}
\tilde h_{t,x_{l-1}}^{(n)}-\bar h_{t,x_{l-1}}^{(n)}=A\tilde h^{(n)}_{t,x_{l}}-A\bar h^{(n)}_{t,x_{l}}.
\end{equation}
Consequently, by Condition \ref{c1} we get
\begin{equation}\label{b4.5}
\left|\tilde h_{t,x_{l-1}}^{(n)}-\bar h_{t,x_{l-1}}^{(n)}\right|\leq\alpha\left|\tilde h_{t,x_l}^{(n)}-\bar h_{t,x_l}^{(n)}\right|.
\end{equation}
Iterating this inequality we obtain
\begin{equation}\label{b4.6}
\left|\tilde h_{t,x_m}^{(n)}-\bar h_{t,x_m}^{(n)}\right|\leq\alpha^{n-m}\left|\tilde h_{t,x_n}^{(n)}-\bar h_{t,x_n}^{(n)}\right|.
\end{equation}

For arbitrary $N, M > n$, we now consider the sets $\{h_{t,x}^{(N)}, x\in V_N\}$ and $\{h_{t,x}^{(M)}, x\in V_M\}$,
$t\in [0,1]$ determined by initial conditions of the form (\ref{b4.1}) for $x\in W_N$ and
$x\in W_M$ respectively and by recursion equations (\ref{eh}). We set $\bar h_{t,x_n}^{(n)}=h_{t,x_n}^{(N)}$, $\tilde h_{t,x_n}^{(n)}=h_{t,x_n}^{(M)}$.
Then inequalities (\ref{b4.6}) imply
$$
\left|h_{t,x_m}^{(N)}-h_{t,x_m}^{(M)}\right|\leq\alpha^{n-m}\left|h_{t,x_n}^{(N)}-h_{t,x_n}^{(M)}\right|\leq 2h^{\max}_0\alpha^{n-m}.
$$

This estimate implies that the sequence $h^{(n)}_{t,x_m}$ satisfies the Cauchy criterion as $n\to\infty$ for a fixed
$m$ and a fixed $t\in [0,1]$; therefore, limit (\ref{b4.2}) exists and is independent of the choice of $h^{(n)}_{t,x_n}$ in (\ref{b4.1}). Because,
by construction, the sets $\{h^{(n)}_{t,x}\}$ satisfy equation (\ref{eh}) before taking the limit, so does $\{h_{t,x}\}$. The uniqueness of $\{h_{t,x}\}$ obviously follows from estimate (\ref{b4.6}).
\end{proof}
A real number $r=r(\pi)$, $0\leq r\leq 1$ can be assigned to the infinite path (see \cite{BG}) and by Theorem \ref{t2} the set $h^{\pi(r)}$ satisfying (\ref{eh}) is uniquely defined.
By the construction of $h^{\pi(r)}$ it is obvious that they are distinct for different
$r\in [0,1]$. We denote by $\nu_r$ the Gibbs measure corresponding to $h^{\pi(r)}$, $r\in [0,1]$.
 One thus obtains uncountable many  Gibbs measures, i.e., we proved the following
\begin{thm}\label{t3} If Conditions \ref{c1} and \ref{c2} are satisfied then for any $r\in [0,1]$ there exists a non-translation-invariant Gibbs measure $\nu_r$ and $\nu_r\ne \nu_l$ if $r\ne l$.
\end{thm}

\subsection{Zachary construction}

In this subsection we adapt Zachary's construction (\cite{Z},\cite[p.251]{Ge}),
which was done for the Ising model, for our model (\ref{e1}) on the Cayley tree.

\begin{con}\label{c3} Assume $K(t,u)>0$ such that the operator A, (\ref{Af}), is invertible.
\end{con}

From (\ref{eh}) we get that
\begin{equation}\label{h1}
h^{\min}(t)\leq h(t,x)\leq h^{\max}(t), \ \ \forall x\in V,
\end{equation}
where
$$h^{\min}(t)=k\ln{\min_{u\in [0,1]}K(t,u)\over \max_{u\in [0,1]}K(0,u)}, \ \ h^{\max}(t)=k\ln{\max_{u\in [0,1]}K(t,u)\over \min_{u\in [0,1]}K(0,u)}.$$

Under Conditions \ref{c2} and \ref{c3} we shall construct a continuum of distinct
functions $h^\zeta_{t,x}$,  which satisfy the functional equation (\ref{eh}), where $\zeta(t)$, is such that
\begin{equation}\label{ze}
h^{\min}(t)<\zeta(t)<h^{\max}(t), \ \ \forall t\in [0,1].
\end{equation}

Take any $\zeta(t)$ with condition (\ref{ze}). Define the sequence
$\zeta_n(t)$, $n\geq 0$ recursively by $\zeta_0(t)=\zeta(t)$,
\begin{equation}\label{tn}
\zeta_n(t)=kA\zeta_{n+1}(t), \ \ n\geq 0.
\end{equation}

Since the operator $A$ is invertible the definition of $\zeta_n(t)$ given by
(\ref{tn}) is unambiguous.

Define the function $h^\zeta_{t,x}$ by $h^\zeta_{t,x}=\zeta_n(t)$ for all $x\in W_n$.

Now we check that this function satisfies the equation (\ref{eh}): for any $x\in V$ there
is $n\geq 0$ such that $x\in W_n$, consequently, $S_k(x)\subset W_{n+1}$ and we have
$$h^\zeta_{t,x}=\sum_{y\in S_k(x)}Ah^\zeta_{t,y}=\sum_{y\in S_k(x)}A\zeta_{n+1}(t)=kA\zeta_{n+1}(t)=\zeta_n(t),$$
i.e., the function $h^\zeta_{t,x}$ satisfies (\ref{eh}) for any $t$ and $\zeta$.

By the construction, distinct functions $\zeta$  define distinct functions $h^\zeta=\{h^\zeta_{t,x}, \, x\in V, t\in [0,1]\}$.
Denote by $\mu^\zeta$ the Gibbs measure which corresponds to the function $h^\zeta$.

Thus we proved the following
\begin{thm}\label{zc}
If Conditions \ref{c2} and \ref{c3} are satisfied, then for any $\zeta$ satisfying (\ref{ze}) there exists a Gibbs measure $\mu^\zeta$ such that $\mu^\zeta\ne \mu^\eta$ if $\zeta\ne \eta$.
\end{thm}

\subsection{Discussions} First our construction (Theorem \ref{t1}) is
adaptation of the ART-construction. In particular, it follows from Theorem \ref{t1}
that if for the model (\ref{e1}) (satisfying the conditions of Theorem \ref{t1})
there are more than one Gibbs measures on a Cayley tree
of order $k_0$ then it has more than one Gibbs measures for any $k\geq k_0$.

Theorem \ref{t3} gives uncountable set of Gibbs measures. Taking any two of these
measures (i.e. corresponding to two values of $t\in [0,1]$) one can  use the argument
of the subsection 2.2. to extend the set
of Gibbs measures.  Zachary's construction is also a way to give uncountable set of Gibbs measures.

It is known that the set of all Gibbs measures of the model (\ref{e1})
is a non-empty, convex and compact subset in the set of all probability measures on $(\Omega,\mathcal B)$
(see \cite[Chapter 7]{Ge}). Therefore it is interesting to know the extreme elements (Gibbs measures)
of the set of all Gibbs measures. Checking extremality of a given Gibbs measure is a difficult problem.
Our constructions of measures  in Theorems \ref{t1}-\ref{zc} are based on known Gibbs measures,
if the known  measures are extreme then measures mentioned in Theorems \ref{t1}-\ref{zc} are extreme too.
In general, the problems of extremality of measures (mentioned in Theorems \ref{t1}-\ref{zc}) remain open.
 Since our analysis related to non-linear integral equations, it seems difficult to give examples
 satisfying Conditions \ref{c1}-\ref{c3}.
 
 \section*{ Acknowledgements}
 We thank all (three) referees for their helpful suggestions.

\end{document}